\documentclass[11 pt]{article}
\usepackage[a4paper, includefoot,
  left=3cm, right=1.5cm,
  top=2cm, bottom=2cm,
  headsep=1cm, footskip=1cm]{geometry}

\usepackage{latexsym,amssymb}
\usepackage{amsmath}
\usepackage{amsthm}
\usepackage{thumbpdf}
\usepackage{graphicx}
\usepackage{subcaption}
\usepackage{hyperref}
\usepackage{siunitx}
\usepackage{cmap}
\usepackage{euscript}

\newcommand{\cJ}{\EuScript{J}}

\newcommand{\rmT}{\mathrm{T}}

\newcommand{\tP}{\mathbb{P}}
\newcommand{\tQ}{\mathbb{Q}}
\newcommand{\tR}{\mathbb{R}}
\newcommand{\tS}{\mathbb{S}}

\newcommand{\tU}{\mathbb{U}}

\newcommand{\tX}{\mathbb{X}}

\newcommand{\bfA}{\mathbf{A}}

\newcommand{\bfC}{\mathbf{C}}

\newcommand{\bfE}{\mathbf{E}}

\newcommand{\bfM}{\mathbf{M}}

\newcommand{\bfP}{\mathbf{P}}
\newcommand{\bfQ}{\mathbf{Q}}

\newcommand{\bfS}{\mathbf{S}}
\newcommand{\bfT}{\mathbf{T}}

\newcommand{\bfX}{\mathbf{X}}

\newcommand{\gA}{\mathfrak{A}}
\newcommand{\gB}{\mathfrak{B}}

\newcommand{\calH}{\mathcal{H}}

\newcommand{\calT}{\mathcal{T}}

\newcommand{\bt}{\begin{theorem}}
\newcommand{\et}{\end{theorem}}
\newcommand{\bl}{\begin{lemma}}
\newcommand{\el}{\end{lemma}}
\newcommand{\bp}{\begin{proposition}}
\newcommand{\ep}{\end{proposition}}
\newcommand{\bc}{\begin{corollary}}
\newcommand{\ec}{\end{corollary}}

\newcommand{\bd}{\begin{definition}\rm}
\newcommand{\ed}{\end{definition}}
\newcommand{\bex}{\begin{example}\rm}
\newcommand{\eex}{\end{example}}
\newcommand{\br}{\begin{remark}\rm}
\newcommand{\er}{\end{remark}}

\newcommand{\btbh}{\begin{table}[!ht]}
\newcommand{\etb}{\end{table}}
\newcommand{\bfgh}{\begin{figure}[!ht]}
\newcommand{\efg}{\end{figure}}

\newcommand{\bea}{\begin{eqnarray*}}
\newcommand{\eea}{\end{eqnarray*}}
\newcommand{\be}{\begin{eqnarray}}
\newcommand{\ee}{\end{eqnarray}}

\def\spaceN{\mathsf{N}}
\def\spaceZ{\mathsf{Z}}
\def\spaceR{\mathsf{R}}
\def\spaceC{\mathsf{C}} 

\def\mmod{\mathop{\mathrm{mod}}}
\def\sspan{\mathop{\mathrm{span}}}
\def\rank{\mathop{\mathrm{rank}}}

\newcommand{\diag}{\mathop{\mathrm{diag}}}

\makeatletter
\def\adots{\mathinner{\mkern2mu\raise\p@\hbox{.}
\mkern2mu\raise4\p@\hbox{.}\mkern1mu
\raise7\p@\vbox{\kern7\p@\hbox{.}}\mkern1mu}}
\newcommand{\l@abcd}[2]{\hbox to\textwidth{#1\dotfill #2}}
\makeatother

\newcommand{\mvec}{\mathop{\mathrm{vec}}}

\newcommand{\Ms}{\mathfrak{M}}
\newcommand{\Ns}{\mathfrak{N}}
\newcommand{\Ls}{\mathfrak{L}}
\newcommand{\Ks}{\mathfrak{K}}
\newcommand{\As}{\mathfrak{A}}

\newcommand{\LsE}{\ell}
\newcommand{\KsE}{\kappa}
\newcommand{\AsE}{\alpha}

\newcommand{\elmat}{\bfE}
\newcommand{\Pmat}{P}

\def\sfH{\mathsf{H}}
\def\sfS{\mathsf{S}}

\newcommand{\pkg}[1]{\textsc{#1}}

\newtheorem{proposition}{Proposition}
\newtheorem{corollary}{Corollary}
\newtheorem{theorem}{Theorem}
\newtheorem{remark}{Remark}
\newtheorem{lemma}{Lemma}
\newtheorem{definition}{Definition}
\newtheorem{example}{Example}
\newtheorem{algorithm}{Algorithm}

\title{\LARGE \bf
Shaped extensions of singular spectrum analysis%
\thanks{This work is supported by the NG13-083 grant of Dynasty Foundation}
}

\author{Alex Shlemov%
\thanks{St.~Petersburg State University, Russia
        {\tt\small shlemovalex@gmail.com}}
~and Nina Golyandina
\thanks{St. Petersburg State University, Russia
        {\tt\small n.golyandina@spbu.ru}}
}

\begin{document}

\maketitle
\thispagestyle{empty}
\pagestyle{empty}

\begin{abstract}
Extensions of singular spectrum analysis (SSA) for processing of non-rectangular images,
which can be also used for analysis of images and time series with gaps, are considered.
A circular version is suggested, which allows application of the method to the data given on a circle
or on a cylinder, e.g. cylindrical projection of a 3D ellipsoid.
The constructed Shaped SSA method with planar or circular topology
is able to produce low-rank approximations for images of complex shapes.
Together with Shaped SSA, a shaped version of the subspace-based ESPRIT method for frequency estimation is developed.
Examples of 2D circular SSA and 2D Shaped ESPRIT are presented.
\end{abstract}

\section{Introduction}
\label{sec:intro}
Singular spectrum analysis
\cite{Broomhead.King1986, Vautard.etal1992, Elsner.Tsonis1996, Golyandina.etal2001, Ghil.etal2002, Golyandina.Zhigljavsky2012}
is a widely used method of time series analysis and image processing, which is based on singular value decomposition of Hankel and block-Hankel matrices. From the viewpoint of low-rank approximation, singular spectrum analysis is related to the problem of structured low rank approximation
\cite{Markovsky2012}.

SSA is closely related to subspace-based methods, where the signal estimation
is performed  with the help of the signal subspace. There is a lot of subspace-based methods of frequency estimation,
for example, the ESPRIT method \cite{Roy.Kailath1989} and its analogues \cite{Barkhuijsen.etal1987, VanHuffel.etal1994}.
The SVD of Hankel-like matrices is also widely used in system identification and model reduction (see \cite{Markovsky2012} for an overview).

Application of SSA is by far not limited to low-rank approximation and includes the problems of smoothing, trend extraction,
decomposition into a sum of trend, periodic and noise components, and many others. In general, SSA does not need the model given in advance
and can serve as a model-building tool \cite{Golyandina.etal2001,Golyandina.Zhigljavsky2012}.

SSA was originally proposed for time series and then was extended to 2D digital images (2D-SSA) \cite{Golyandina.Usevich2010,Golyandina.etal2013}. Related is the subspace-based 2D-ESPRIT method for estimation of 2D frequencies (see \cite{Rouquette.Najim01Estimation}). 2D-SSA and related subspace based methods find applications in texture analysis \cite{monadjemi04-towards}, seismology \cite{trickett08conf-f}, spatial gene expression data \cite{Holloway.etal2011}, medical imaging \cite{Shin.etal13MRiM-Calibrationless}, etc., and are gaining increasing popularity.

The SSA and 2D-SSA methods deal with series and rectangular images and use intervals and rectangles as moving windows.
This can limit applications of the methods; for example, the methods hardly process  circular-shaped images,
images with gaps, and so on. Due to this reason, in \cite{Golyandina.etal2013} Shaped 2D-SSA
was proposed, which can deal with images of arbitrary shape and can use arbitrary window shapes.

The SSA-based methods are supported by an effective implementation in the freely available \pkg{R}-package \pkg{Rssa} \cite{Korobeynikov.etal2013}. The implementation uses the fact that various SSA versions can be considered as particular cases of Shaped SSA.  This allows a unified implementation for ordinary SSA, multivariate SSA (MSSA),
2D-SSA, SSA for analysis of objects with gaps. The  details can be found in \cite{Golyandina.etal2013, Golyandina.Korobeynikov2013}.

The contribution of this paper is twofold. Firstly, we consider an extension of Shaped SSA
to circular topologies, such as cylindrical and toroidal. This extension is motivated by problems in gene expression data analysis  \cite{Holloway.etal2011}, where the data are obtained by cylindric projections of 3D surfaces. We discuss theoretical properties and effective implementation of this extension.
Secondly, we propose the Shaped ESPRIT method of parameter estimation for shaped images. The paper~\cite{Golyandina.etal2013} does not contain the description of Shaped ESPRIT and we compensate this here.
Several theoretical aspects are discussed. An implementation of both contributions is included in the \pkg{Rssa} package \cite{Korobeynikov.etal2013}.

In Section~\ref{sec:shaped_algo}, we describe the common scheme of SSA methods, which differ mainly in the
embedding step, and introduce the embedding operator for Shaped SSA, which includes both
planar and circular cases.
In Section~\ref{sec:implementation}, we discuss efficient implementation aspects and show that the methods developed for
the planar case are naturally extended to the circular case.
Section~\ref{sec:shaped_rank} contains several propositions about the Shaped-SSA rank of infinite arrays
and the common term of infinite arrays of finite Shaped-SSA rank.
In Section~\ref{sec:shaped_esprit}, we consider a common scheme of ESPRIT-like methods, then
present the algorithm of ESPRIT for the shaped circular case and prove its correctness.
In Section~\ref{sec:examples}, we consider two examples, for demonstration of the shaped version of 2D-ESPRIT and for
demonstration of the circular version of 2D-SSA.
Section~\ref{sec:conclusions} closes the paper with conclusions and further remarks.

\section{Shaped SSA algorithm}
\label{sec:shaped_algo}

\subsection{Structure of SSA-like algorithms}
We will follow the common structure of SSA-like algorithms described in \cite{Golyandina.etal2013}.
This common structure consists of the embedding, decomposition, grouping and reconstruction steps.
The input of an SSA-like algorithm is an object $\tX$ to be decomposed (time series, multivariate time series, image, etc.). The result of SSA-like algorithms is a decomposition of $\tX$  on
identifiable components in the form $\tX = \tX_1 + \dots + \tX_s$.

\smallskip
\textbf{1. Embedding step}:  construction of the trajectory matrix $\bfX=\calT (\tX)\in \sfH$,
where $\sfH$ is a space of structured  Hankel-like matrices. The structure of the matrix $\bfX$ (and the space $\sfH$) depends on the object to be decomposed  and on the \emph{window}, which is the main algorithm parameter.
In a sense, the window size reflects the resolution of the method, i.e., larger windows lead to more detailed decompositions.

\smallskip
\textbf{2. Decomposition step}: Singular Value Decomposition (SVD) of the trajectory matrix $\bfX=\sum_{i=1}^d \sqrt{\lambda_i} U_i V_i^\top = \sum_{i=1}^d \bfX_i$. Here $(\sqrt{\lambda_i},U_i,V_i)$ are so called eigentriples (abbreviated as ET) and consist of singular values, left and right singular vectors of $\bfX$.
Also, $\lambda_i$ and $U_i$ are eigenvalues and eigenvectors of $\bfX \bfX^\top$, $\{U_i\}_{i=1}^d$ forms an orthonormal basis of the column space
of $\bfX$, while $\{V_i\}_{i=1}^d$ forms an orthonormal basis of the row space. Vectors $V_i$ are also called factor vectors.

\smallskip
\textbf{3. Grouping step}:
partition $\{1,\ldots,d\} = \bigsqcup_{j=1}^s I_j$ and grouping of summands in the SVD decomposition to obtain a grouped matrix decomposition $\bfX=\sum_{j=1}^s \bfX_{I_j}$, where $\bfX_I= \sum_{k\in I} \bfX_k$.
The grouping with  $I_j=\{j\}$ is called elementary. The aim of this step is to group the SVD components to obtain an interpretable decomposition of the initial
object. This can be performed by means of analysis of eigentriples.

\smallskip
\textbf{4. Reconstruction step}: decomposition of the initial object $\tX = \tX_1 + \ldots + \tX_s$,
where $\tX_j = \calT^{-1} \calH (\bfX_{I_j})$, $\calH$ is the operator of projection
on the space $\sfH$ (e.g., hankelization in the 1D case); $\calH(\bfX_I)=\sum_{i\in I}  \calH(\bfX_i)$ holds.

In this paper, we consider the Shaped SSA algorithm for shaped arrays. In different basic variants of SSA, all the steps except Embedding are the same; therefore, we should describe the embedding step only (although, a variety of modifications of SVD step can be also considered, see e.g. \cite{Shlemov.Golyandina2014}).
The embedding step for Shaped SSA is described in detail in Section~\ref{sec:embedding_shaped}. For simplicity, we consider only the 2-dimensional Shaped SSA. Multidimensional generalizations can be obtained analogously.

\subsection{Comments on the SSA-like algorithms}
For one-dimensional SSA, the theory related to choice of window sizes and grouping can be found
in \cite{Golyandina.etal2001,Golyandina.Zhigljavsky2012}. For multidimensional extensions this theory can be generalized;
however, it is much more complex and not well elaborated.
The main notion is the separability of components of the initial object. The separability theory answers the question
if the SVD expansion at Decomposition step allows the grouping that provides a desirable decomposition of the initial object
(at least, approximately); e.g. decomposition on such identifiable components as trend and oscillations, signal and noise, and so on.
Fortunately, the conditions of approximate separability are not too restrictive and many of series components can be
separated by SSA.

Thus, to perform SSA analysis for component extraction, one should choose window sizes based on the SSA theory
and then choose the set $I$ of eigentriple numbers (or the set of eigenvectors $\{U_i, i \in I\}$).
This set of eigenvectors can provide information about the object parametric model and also
allows the reconstruction of the extracted component, e.g., of a trend, a periodic component or a signal.

The model that underlies the most of the theory of SSA-like methods is associated with objects with rank-deficient trajectory matrices.
This model is related to the objects satisfying linear recurrence relations (e.g., $x_n=\sum_{i=1}^r a_i x_{n-i}$
for one-dimensional time series).

SSA-like methods are called non-parametric or free of model, since they do not use the model in the algorithm, what is very convenient
for the exploratory analysis, and also they allow to analyze the objects that satisfy the model only locally or
do not satisfy it at all; for example, SSA-like methods can be used for smoothing and filtering \cite{Golyandina.Zhigljavsky2012}.

\subsection{Shapes and actions with them}
\label{ssec:shapesandactions}

Formally, we call \emph{shape} $\gB$ a finite non-empty subset of $\{1,\ldots,T_x\}\times \{1,\ldots,T_y\}$,
where  $T_x, T_y \in \spaceN \cup \{\infty\}$, $\spaceN$ is the set of natural numbers.
The values $T_x$ and $T_y$ characterize the topology of the set containing this shape.
Finiteness of $T_x < \infty$ ($T_y < \infty$) means that the topology is circular and
the shapes are cyclic by the $x$-coordinate (the $y$-coordinate) with period $T_x$ ($T_y$).

A\emph{ $\gB$-shaped array }is a partially indexed array $\tX = \tX_\gB = (x_{(i, j)})_{(i, j) \in \gB}$.
Let us denote the space of $\gB$-shaped arrays as $\spaceR^\gB$.
Obviously, there is an isomorphism $\spaceR^\gB \sim \spaceR^B$,
where $B = |\gB|$ is the cardinality of the set $\gB$.
This isomorphism is not unique, since the elements in the shape can be ordered in different ways.
However, the result of the algorithm does not depend on the particular order.

For convenience, we consider the lexicographical order, which fixes the isomorphism
\begin{equation}
\label{eq:iso}
\cJ_\gB:\,\spaceR^\gB \mapsto \spaceR^B.
\end{equation}
We call $\cJ_{\gB}$ \emph{vectorization} and its inverse $\cJ_{\gB}^{-1}$ \emph{shaping}.

Next, we introduce the operation of addition in the considered topology
for two pairs of indices $\LsE = (\LsE_x, \LsE_y)$, $\KsE = (\KsE_x, \KsE_y)$.
For convenience, we omit the parameters $(T_x, T_y)$ and write $\oplus$ instead of $\oplus^{(T_x, T_y)}$:
\begin{equation*}
  \LsE \oplus \KsE =
  \big((\LsE_x + \KsE_x - 2) \mmod T_x + 1, (\LsE_y + \KsE_y - 2) \mmod T_y + 1\big).
\end{equation*}
Here $\mmod$ denotes the remainder of the integer division. Formally, $a \mmod \infty = a$ for any $a$.
Note that $-2$ and $+1$ in the definition of $\oplus$ is the consequence of the index enumeration starting from 1 (such that $\{1\}\oplus\{1\} = \{1\}$).

For two shapes $\gA$ and $\gB$, we modify the definition of the Minkowski sum in the following way:
\begin{equation*}
  \gA \oplus \gB = \{\alpha \oplus \beta \mid \alpha \in \gA, \beta \in \gB \}.
\end{equation*}

\subsection{Embedding step}\label{sec:embedding_shaped}
Now we are ready to define the embedding operator for the Shaped SSA algorithm.
\subsubsection{Input data and parameters of the embedding}
The input data are the topology characteristics $(T_x, T_y)$,
the shape $\Ns$, where $\Ns \subset \{1, \dots, T_x \} \times \{1, \dots, T_y\}$,
the $\Ns$-shaped array $\tX \in \spaceR^\Ns$.
The parameter of the algorithm is a \emph{window shape} $\Ls \subset \Ns$.
It is convenient to consider
the window shape in the form $\Ls = \{\LsE_1, \dots, \LsE_L\}$,
where $L=|\Ls|$ and $\LsE_i \in \spaceN^2$ are ordered in lexicographical order.

For each $\KsE \in \Ns$, we define a shifted $\Ls$-shaped subarray as
$\tX_{\Ls \oplus \{\KsE\}} = (x_\alpha)_{\alpha \in \Ls \oplus \{\KsE\}}$.
The index $\KsE$ is a position of the origin for the window. Consider the set of all possible
origin positions for $\Ls$-shaped windows:
\begin{equation*}
  \Ks = \{\KsE \in \Ns \mid \Ls \oplus \{\KsE\} \subset \Ns\}.
\end{equation*}
We assume that $\Ks$ is written as $\Ks = \{\KsE_1, \dots, \KsE_K\}$, where $K=|\Ks|\neq 0$ and $\KsE_j$ are ordered in lexicographical order.

If the shapes $\Ns$ and $\Ls$ are rectangles, then $\Ks$ is also rectangular and we call the version of Shaped SSA \emph{rectangular}.

\subsubsection{Embedding operator}
The trajectory matrix $\bfX$ is constructed by the embedding operator $\calT_\mathrm{ShSSA}: \spaceR^\Ns \to \spaceR^{L \times K}$
\begin{equation*}
\calT(\tX) = \calT_\mathrm{ShSSA}(\tX) := \bfX = [X_1, \dots, X_K],
\end{equation*}
where the columns
\begin{equation*}
  X_j = (x_{\LsE_i \oplus \KsE_j})_{i=1}^L
\end{equation*}
are vectorizations of the shaped subarrays $\tX_{\Ls \oplus \{\KsE_j\}}$.

The embedding operator $\calT_\mathrm{ShSSA}$ is obviously linear.
Denote its range as $\sfH^{\Ls \times \Ks} \subset \spaceR^{L \times K}$.
The subspace  $\sfH^{\Ls \times \Ks}$ consists of structured matrices, which we will call \emph{quasi-Hankel}. (In fact, they are generalizations of quasi-Hankel matrices from \cite{Mourrain.Pan2000}.)
If the operator $\calT_\mathrm{ShSSA}$ is injective, then it sets the isomorphism of the spaces
$\spaceR^\Ns$ and $\sfH^{\Ls \times \Ks}$.

\begin{remark}
The embedding operator $\calT_\mathrm{ShSSA}$ is injective if and only if $\Ls \oplus \Ks = \Ns$, i.e.,
if each point of the initial shaped array $\tX$ can be covered by a window of shape $\Ls$.
\end{remark}

If there are uncovered points, we can remove them and consider only the decomposition of the restricted array
$\tX' = (\tX)_{\Ns'}$, where $\Ns' = \Ls \oplus \Ks$.
Hereafter we will suppose that $\Ns = \Ls \oplus \Ks$, i.e., all the points are covered and the operator $\calT_\mathrm{ShSSA}$ is injective.

\subsection{Particular cases}
As shown in \cite{Golyandina.etal2013}, many variants of SSA can be considered as special cases of Shaped SSA. Basic SSA corresponds to the shapes
\begin{equation*}
  \Ns = \{1, \dots, N\} \times \{1\}, \qquad \Ls = \{1, \dots, L\} \times \{1\}.
\end{equation*}
The shaped arrays $\tX \in \spaceR^\Ns$ in this case are time series, and the set  $\sfH^{\Ls \times \Ks}$  is the set of ordinary Hankel matrices. SSA for time series with  missing values is obtained by removing the indices of missing elements from the shape $\Ns$.

Multivariate SSA (MSSA) for multivariate time series of possibly different lengths
can be considered as Shaped SSA for a shaped 2D array consisting of stacked series
and a rectangular window shape $\Ls = \{1, \dots, L\} \times \{1\}$.
Then
the space $\sfH^{\Ls \times \Ks}$ is the space of concatenated Hankel matrices.

The 2D-SSA algorithm \cite{Golyandina.Usevich2010} corresponds to Shaped SSA with
\begin{gather*}
  \Ns = \{1, \dots, N_x\} \times \{1, \dots, N_y\},\\
  \Ls = \{1, \dots, L_x\} \times \{1, \dots, L_y\}.
\end{gather*}
The shaped arrays $\tX \in \spaceR^\Ns$ are rectangular images, the windows $\Ls$ are rectangular, and the set $\sfH^{\Ls \times \Ks}$ is the set of \emph{Hankel-block-Hankel} matrices. For more details on different SSA variants as special cases of Shaped SSA, we refer the reader to \cite[Sec. 5.3]{Golyandina.etal2013}.

All these SSA versions can be extended by allowing circular topology. In order to distinguish between different variants, we use the following terminology. If both $T_x$ and $T_y$ are infinite,
then we call the SSA version \emph{planar} (see Figure~\ref{fig:shaped_demo}, left); otherwise the term \emph{``circular''} is used. The case, when only one of $T_x$ and $T_y$ equals infinity, can be named
\emph{cylindrical} (see Figure~\ref{fig:shaped_demo}, center), while the case, when both of them are finite, is \emph{toroidal} (see Figure~\ref{fig:shaped_demo}, right).

\begin{figure}
  \centering
  \vspace{0.3cm}
  \begin{subfigure}{.32\columnwidth}
    \centering
    \includegraphics[width=0.6\columnwidth]{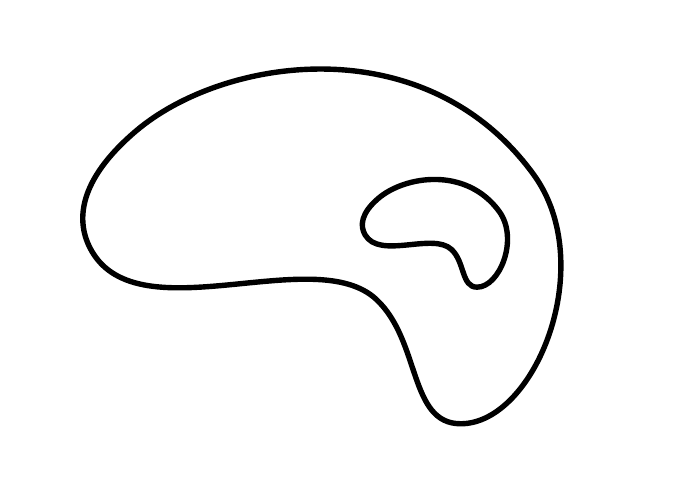}
  \end{subfigure}
  \begin{subfigure}{.32\columnwidth}
    \centering
    \includegraphics[width=0.6\columnwidth]{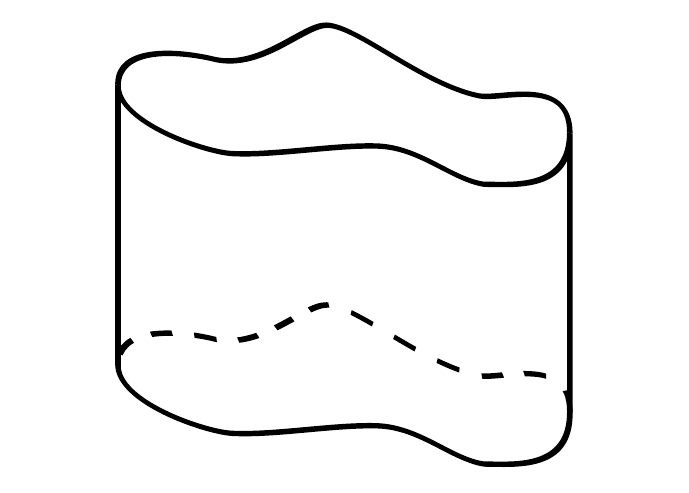}
  \end{subfigure}
  \begin{subfigure}{.32\columnwidth}
    \centering
    \includegraphics[width=0.6\columnwidth]{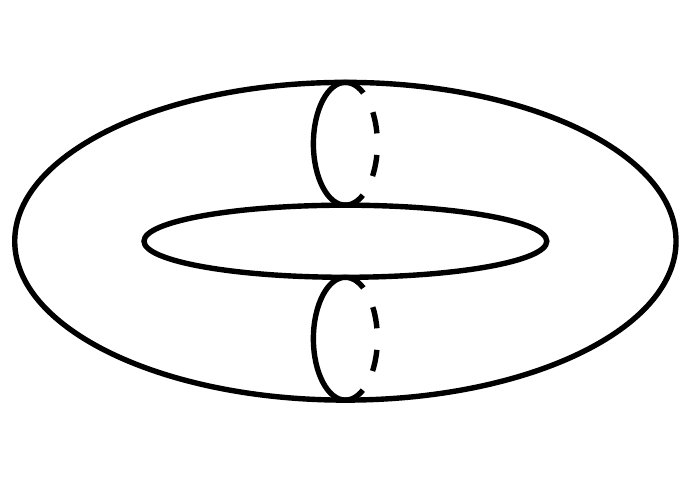}
  \end{subfigure}
  \caption{Left: shaped image and window; center: cylindrical topology; right: toroidal topology}
  \vspace{-0.5cm}
  \label{fig:shaped_demo}
\end{figure}

\section{Implementation issues}
\label{sec:implementation}
\subsection{Approaches to fast implementation}
The fast implementation of SSA-related methods, which was suggested in \cite{Korobeynikov2010} and is used in the \pkg{Rssa} package
\cite{Korobeynikov.etal2013},
relies on the following approaches:
\begin{enumerate}
  \item The truncated SVD calculated by Lanczos methods \cite[Ch. 9]{Golub.VanLoan1996} is used, since only a number of leading
  SVD components correspond to the signal and therefore are used at Grouping step of the SSA algorithm.
   \item Lanczos methods do not use an explicit representation of a decomposed matrix $\bfA$.
   They need only the results of multiplication of $\bfA$ and $\bfA^\top$ on some vectors.
   Due to special Hankel-type structure of $\bfA$ in the SSA-related algorithms,
   multiplication by a vector can be implemented very efficiently with the help of Fast Fourier Transform (FFT).
     Fast SVD algorithms are implemented in the \pkg{R}-package \pkg{svd} \cite{Korobeynikov2013} in such a way that
     their input is the function of a vector which yields the result of fast multiplication of the vector by the trajectory matrix.
     Therefore, the use of \pkg{svd} in \pkg{Rssa} allows
     a fast and space-efficient implementation of the SSA-related algorithms.
   \item At Reconstruction step, hankelization or quasi-hankelization of a matrix of rank $1$ stored as  $\sqrt{\lambda}UV^\rmT$
  can be written by means of the convolution operator and therefore is also can be effectively implemented.
\end{enumerate}
The overall complexity of the computations is $O(r N \log (N))$, where $N = |\Ns|$ is the number of elements in the shaped array and $r$ is the number of considered eigentriples, see details in \cite{Korobeynikov2010} and in \cite{Golyandina.etal2013}. This makes the computations feasible for large data sets and large window sizes.  For example, the case of images of size $299 \times 299$ and a window size $100\times 100$ can be handled easily, whereas the conventional algorithms (e.g., the full SVD \cite{Golub.VanLoan1996}) are impractical, because the decomposed matrix is of size $10^4 \times 4\cdot 10^4$. Using larger window sizes is advantageous, since, for example, separability of signal from noise (in the ``signal+noise'' scenario) can be significantly improved \cite{Golyandina2010}.

These approaches were implemented for 1D, 2D-SSA
and planar version of Shaped 2D-SSA in the version of the \pkg{Rssa} package described in \cite{Golyandina.etal2013}.

\subsection{Fast implementation for circular topologies}
In this section, we show how the same approaches can be used for Shaped SSA with circular topologies. As
 shown next, it is not difficult to extend the expressions in \cite[Sec. 6]{Golyandina.etal2013} to the general case. This allows us to incorporate circular topologies in the new version of the \pkg{Rssa} package \cite{Korobeynikov.etal2013}.

\subsubsection{Use of the Lanczos methods}
Since the circular version differs by the form of the embedding operator only,
the Lanczos methods for the truncated SVD can accelerate the circular version of SSA algorithms in the same way
as they do it for the planar case.

\subsubsection{Fast multiplication of a vector on the trajectory matrix}
The fast implementation in the planar case is based on the representation of the trajectory matrix with the help of
a Hankel or Hankel-block-Hankel circulant. This was suggested for 1D case in \cite{Korobeynikov2010}. It was shown in \cite{Golyandina.etal2013}
that the same approach is valid for the shaped 2D case.

The first step is the embedding of the initial array into the Hankel-block-Hankel circulant $\bfC_{\rm HbH}$.
It is performed in a straightforward way for the rectangular case. For the shaped case, the $\Ns$-shaped array  is
padding by zeros to the rectangular array of sizes $N_x$ and $N_y$  such that
\begin{equation}
  \label{eq:NxNyCorrectness}
  \Ns \subset \{1, \dots, N_x\} \times \{1, \dots, N_y\}.
\end{equation}
In \cite{Golyandina.etal2013} the $\Ls$-trajectory matrix of a shaped 2D array $\tX$ is presented in the form
\begin{equation*}
  \bfX = \calT_{\rm ShSSA}(\tX) = P^\rmT_\Ls \bfC_{\rm HbH} P_\Ks,
\end{equation*}
where $P_\Ls\colon \spaceR^\Ls \to \spaceR^{N_x \times N_y}$ and $P_\Ks\colon \spaceR^\Ks \to \spaceR^{N_x \times N_y}$.
A shape  $\As = \{\AsE_1, \dots, \AsE_A\}$ is defined as
\begin{equation}
\label{eq:circ_rep}
  \Pmat_{\mathfrak{A}} = [\mvec(\elmat_{\AsE_1}) : \cdots : \mvec(\elmat_{\AsE_A})],
\end{equation}
$\elmat_{\AsE_i}$ is the $\As$-shape with $1$ on the $\AsE_i$th entry.
Then multiplication by $\bfX$ is reduced to multiplication by $\bfC_{\rm HbH}$, which is performed effectively by means of 2D FFT;
the projectors $P_\Ls$ and $P_\Ks$ are trivially calculated.

Let us consider the peculiarities of the circular case.
First, in the planar case $N_x$ and $N_y$ have only low bounds given by \eqref{eq:NxNyCorrectness}.
In the circular case by $x$ we should take $N_x = T_x$ and in the circular case by $y$ it is necessary to set $N_y = T_y$.
Then the circular topology transfers to the cyclic structure of the circulant.
Note that the projector $\Pmat_{\mathfrak{L}}$ does not depend on either planar or circular case is considered,
while the matrix $\Pmat_{\mathfrak{K}}$ is generally larger for the circular case.

\subsubsection{Rank-one quasi-Hankelization}
Extension of fast Hankelization is also quite straightforward.
The formulas from \cite{Korobeynikov.etal2013} are valid with appropriate choice of $N_x$ and $N_y$ equal
to $T_x$ and/or $T_y$ in the circular case.

\section{Rank of shaped arrays}
\label{sec:shaped_rank}
Following the definition of finite-rank time series \cite{Golyandina.etal2001} and finite-rank 2D arrays
\cite{Golyandina.Usevich2010}, let us generalize the notion of rank to the shaped and circular shaped cases.
Note that generally the theory of infinite arrays of finite rank is closely related to multidimensional arrays satisfying
linear recurrence relations.

\begin{definition}
The \emph{$\Ls$-trajectory space} $\sfS^\Ls$ of a shaped array $\tS \in \spaceR^\Ns$ is defined as
\begin{equation*}
  \sfS^\Ls(\tS) = \sspan \{ (\tS)_{\Ls \oplus \{\KsE\}}  \}_{\KsE \mid \Ls \oplus \{\KsE\} \subset \Ns}.
\end{equation*}
\end{definition}

The introduced trajectory space corresponds to the column space of the trajectory matrix $\bfS=\calT_\mathrm{ShSSA}(\tS)$.
Due to the isomorphism $\cJ_\Ls$, $\{\tP_1,\ldots,\tP_r\}$ is a basis of $\sfS^\Ls(\tS)$ if and only if
 $\{P_1,\ldots,P_r\}$, where $P_k=\cJ_\Ls(\tP_k)$, is a basis of the column space of $\bfS$.
 Let $\{Q_1,\ldots,Q_r\}$  be a basis of the row space of $\bfS$, that is, of the row trajectory space.
Then $\{\tQ_1,\ldots,\tQ_r\}$, where $\tQ_k=\cJ_\Ks^{-1}(Q_k)$, forms the basis of a space, which can be denoted as $\sfS^\Ks(\tS)$.
Thus, we use the terms \emph{column}/\emph{row} \emph{trajectory spaces} and \emph{column}/\emph{row} \emph{shaped trajectory spaces,}
 depending on the context.
\begin{definition}
The \emph{$\Ls$-rank} of a shaped array $\tS$ is defined as dimension of its $\Ls$-trajectory subspace:
\begin{equation*}
  \rank\nolimits_\Ls \tS = \dim \sfS^\Ls(\tS) = \rank \bfS.
\end{equation*}
\end{definition}

\subsection{Ranks of finite and infinite arrays}
In order to describe the general form of arrays of finite rank, we consider infinite arrays (and their trajectory spaces) in both sides for both dimensions (i.e., $\Ns=\spaceZ \times \spaceZ$).
\begin{definition}
Infinite (in both sides) array $\tS_\infty$ is called the \emph{array of finite shaped rank}
if $r = \max\limits_\Ls \rank_\Ls \tS_\infty$, where maximum is taken over $\Ls \subset \{1, \dots, T_x\} \times \{1, \dots, T_y\}$,
is finite.
We will write $\rank \tS_\infty = r$.
\end{definition}

\begin{remark}
The shaped rank for an infinite array is equal to its rectangular rank (when the maximum is taken over rectangular shapes),
since a sequence of shapes that contain increasing
rectangular shapes and are contained in the rectangular shapes can be easily constructed. Therefore, we will say
about arrays of finite rank omitting ``shaped''.
\end{remark}

In practice, the following remarks are important.
\begin{remark}
\label{rem:suffsquare}
If $\tS_\infty$ is an infinite array of finite rank $r$, $\tS = (\tS_\infty)_\Ns$ is a sufficiently large
finite subarray of $\tS_\infty$ and $\Ls$ is a sufficiently large shape, then
\begin{equation}
\label{eq:rank_equality}
  \rank\nolimits_\Ls \tS = \rank\nolimits_\Ls \tS_\infty = \rank \tS_\infty = r,
\end{equation}
Let $\Ks = \{\KsE \mid \Ls \oplus \{\KsE\} \subset \Ns\}$.
The sufficient condition for \eqref{eq:rank_equality}
 (see, e.g. \cite{Golyandina.Usevich2010} for the proof)
is that both $\Ls$ and $\Ks$ shapes
contain at least one $r \times r$ square. More formally, it is sufficient that there exist a 2D index $\alpha=(l, n)$
such that $\{1,\ldots, r\} \times \{1,\ldots, r\} \oplus \{\alpha\} \subset \Ls$; the same should be valid for $\Ks$.
\end{remark}

\subsection{Infinite arrays of finite rank: planar case}

\begin{proposition}
\label{prop:2dfiniterankmodel}
Let $T_1 = T_2 = \infty$.
Then an infinite array $\tS_\infty$ of finite rank has the form
\begin{equation}
  \label{eq:2dfiniterankmodel}
  (\tS_\infty)_{l, n} = \sum_{k = 1}^s P_k(l, n) \mu_k^l \nu_k^n,
\end{equation}
where $\mu_k, \nu_k \in \spaceC$, the pairs $(\mu_k, \nu_k)$ are different,
$P_{k}(l, n)$ are complex polynomials of $l$ and $n$;
such representation is unique up to summands order.
\end{proposition}
This important fact is well known starting from  \cite[\S 20.20]{Kurakin.etal95JoMS-Linear}.
Note that the rank of the array $\tS_\infty$ given in \eqref{eq:2dfiniterankmodel} does not determined only by degrees of polynomials $P_k$ and is
hardly calculated (see~\cite{Golyandina.Usevich2010}).

An important special case is $(\tS_\infty)_{l,n}=\sum_{k=1}^r A_k \mu_k^l \nu_k^n$, which is widely used in signal processing. In this case, $\rank \tS_\infty$ is equal to the number of different  $(\mu_k, \nu_k)$-pairs.

\subsection{Infinite arrays of finite rank: circular case}
It follows from the definition that in the $T_x$-circular case only arrays that are $T_x$-periodic by the $x$-axis can be of finite rank.
Thus, we obtain two propositions, which characterize the cylindrical and toroidal cases.

\begin{proposition}
Let $T_x < \infty$, $T_y = \infty$.
Then an infinite array $\tS_\infty$ of finite rank $r$  has the form
\begin{equation}
  \label{eq:cylfiniterank}
  (\tS_\infty)_{l, n} = \sum_{k = 1}^s P_k(n) \mu_k^l \nu_k^n,
\end{equation}
where $\mu_k, \nu_k \in \spaceC$, the pairs $(\mu_k, \nu_k)$ are different,
$\mu_k^{T_x} = 1$,
$P_{k}(n)$ are polynomials of $n$ of degree $p_k$, $\sum_k (p_k+1)=r$;
such representation is unique up to summands order.
\end{proposition}
\begin{proposition}
Let $T_x < \infty$, $T_y < \infty$.
Then an infinite array $\tS_\infty$ of finite rank $r$ has the form
\begin{equation}
  \label{eq:cylfiniterank2}
  (\tS_\infty)_{l, n} = \sum_{k = 1}^r A_k \mu_k^l \nu_k^n,
\end{equation}
where $\mu_k, \nu_k \in \spaceC$, the pairs $(\mu_k, \nu_k)$ are different,
$\mu_k^{T_x} = 1$, $\nu_k^{T_y} = 1$, $A_k$ are some constants;
such representation is unique up to summands order.
\end{proposition}

\begin{remark}
\label{rem:cyclic_f}
In the toroidal case, $T_x < \infty$, $T_y < \infty$, an array of finite rank $r$ can be presented in the form
of sum of $r$ components of the Discrete Fourier Transform.
\end{remark}

\section{Shaped ESPRIT}
\label{sec:shaped_esprit}
The well known ESPRIT \cite{Roy.Kailath1989} and 2D-ESPRIT \cite{Rouquette.Najim01Estimation} methods are used for estimation of parametes in the model  $\tX=\tS+\tR$, where $\tR$ is noise and
$\tS$ is the signal of finite rank, which is of the form $(\tS)_{l}=\sum_{k=1}^r A_k \mu_k^l$ (for ESPRIT) and $(\tS)_{l,n}=\sum_{k=1}^r A_k \mu_k^l \nu_k^n$ (for 2D-ESPRIT).

Next, we briefly describe the general scheme of ESPRIT-like methods, on the example of 2D-arrays.
\subsubsection{Estimation of the signal subspace}
We assume that the signal is at least approximately separated from the residual. In this case, $r$ left singular vectors $\{U_{i_1}, \dots, U_{i_r}\}$ of the trajectory matrix (and corresponding shaped arrays $\{\tU_{i_1}, \dots, \tU_{i_r}\}$, where $\tU_{i_k} = \cJ_\Ls^{-1}(U_{i_k})$)  approximate the column signal subspace. (In signal processing, typically, first $r$ components are chosen, i.e. $ \{i_1, \dots, i_r\} = \{1,\ldots,r\}$.)
\subsubsection{Construction of shifted matrices}
The algorithm uses the notion of shifts for shaped arrays from
the space spanned by $\{\tU_{i_1}, \dots, \tU_{i_r}\}$. Directions of shifts corresponds to the
directions in the initial space. Thus, four matrices $\bfP_x$, $\bfQ_x$, $\bfP_y$, $\bfQ_y$ are constructed, which columns span the shifted subspaces. For 2D-ESPRIT this step was already considered in
\cite{Rouquette.Najim01Estimation}. For Shaped ESPRIT, we define it in Section~\ref{sec:sh_construct}.
\subsubsection{Construction of shift matrices}
 The relation between the shifted subspaces allows to estimate the so called shift matrices of order $r \times r$ as the approximate solutions of $\bfP_x \bfM_x \approx \bfQ_x$ and $\bfP_y \bfM_y \approx \bfQ_y$. The method of estimation is based on least squares (LS-ESPRIT) or total least squares (TLS-ESPRIT).

\subsubsection{Estimation of the parameters}
The eigenvalues of the shift matrices provide estimates of the parameters: the eigenvalues of the $x$-direction shift matrix provide
estimates of $\mu_k$, while the eigenvalues of the $y$-direction shift matrix produce estimates of $\nu_k$. However, the method should
provide estimates of the pairs $(\mu_k,\nu_k)$.
There are different approaches for pair estimation: simultaneous diagonalization of shift matrices (see \cite{Rouquette.Najim01Estimation}) 
and a pairing of independently calculated eigenvalues of two shift matrices (the latter method is called 2D-MEMP \cite{Rouquette.Najim01Estimation}
and is improved in \cite{Wang.etal05IToSP-Comments}).


\subsection{Construction of shifted matrices}
\label{sec:sh_construct}
Let us introduce the shaped version of the ESPRIT method and describe its specific features.
A difference of the shaped version from the rectangular planar one consists in a special shift construction for shapes at the second step
of the ESPRIT scheme. (Other steps are exactly the same as in
the planar versions.) Also, the circular topology influences the algorithm.

\begin{algorithm}[Construction of shifted matrices]
\label{alg:ESPRIT}

\noindent
\textit{Input}: \\
1. The topology of the initial 2D array parameterized by $(T_x, T_y)$; \\
2. the 2D shape $\Ls$ equipped with the topology of the initial object;\\
3. a basis $\tU_1, \dots, \tU_r \in \spaceR^\Ls$
of the estimated signal subspace.
For example, $\{\tU_k\}_{k = 1}^r$ are $\Ls$-shaped eigenvectors from the chosen eigentriples
of the trajectory matrix identified as related to an image component that one wants to analyze.

\smallskip
\noindent
\textit{Output}: the matrices $\bfP_x$, $\bfQ_x$, $\bfP_y$, $\bfQ_y$.

Let us show how to construct the $x$-shifted matrix, that is, the matrix
obtained by the shift along the $x$ axis; the case of the shift along the $y$ axis
is analogous.

1. Find the subset of points from $\Ls$ that have adjacent points of $\Ls$ from the right in the $x$-direction:
\begin{equation}
\label{eq:Msx}
  \Ms_x = \{ \LsE \in \Ls \mid \LsE \oplus (2, 1) \in \Ls  \}
\end{equation}
Recall that the sum $\oplus$ with 1 corresponds to no shift.

2. Consider the shaped eigenarrays restricted on $\Ms_x$ and on $\Ms_x \oplus \{(2, 1)\}$:
\begin{equation}
  \label{eq:tPtQ}
  \tP_k = (\tU_k)_{\Ms_x}, \quad \tQ_k = (\tU_k)_{\Ms_x \oplus \{(2, 1)\}},\  k=1,\ldots,r.
\end{equation}

3. Vectorize
 the arrays $\tP_k$ and  $\tQ_k$  into $P_k$ and $Q_k$ correspondingly.

4. Stack the obtained vectors to the matrices
\begin{equation*}
  \bfP_x= [P_1, : \dots : P_r], \quad \bfQ_x = [Q_1 : \dots : Q_r].
\end{equation*}
\end{algorithm}

\smallskip
Thus, the shifted matrices are constructed and the shift matrices $\bfM_x$ and $\bfM_y$ can be estimated from the relations
\begin{equation}
  \label{eq:ESPRITapprox}
  \bfP_x \bfM_x \approx \bfQ_x, \quad \bfP_y \bfM_y \approx \bfQ_y.
\end{equation}

\begin{remark}
Note that for the toroidal rectangular case the application of ESPRIT has little sense, since the estimation
of frequencies of a finite-rank toroidal signal can be performed by Fourier analysis, see
Remark~\ref{rem:cyclic_f}.
\end{remark}

\begin{remark}
We have considered the version of ESPRIT based on the column space.
The version based on the row space can be also considered.
\end{remark}

\subsection{Conditions of the algorithm correctness}
Let a $\Ns$-shaped array $\tS$ of rank $r$ have the common term in the form $(\tS)_{l,n} = \sum_{k = 1}^r A_k \mu_k^l \nu_k^n$
for a set of different pairs $\{(\mu_k, \nu_k)\}_{k = 1}^r$,
$\mu_k, \nu_k \in \spaceC$, and noise $\tR$ is zero.
The following proposition forms the basis of the ESPRIT-type methods (see e.g. \cite{Rouquette.Najim01Estimation})
and provides the conditions, when Algorithm~\ref{alg:ESPRIT} and the ESPRIT scheme as a whole are correct, that is, the equations~\eqref{eq:ESPRITapprox}
 have exact solutions
and the constructed shift matrices produce exactly the pairs $(\mu_k, \nu_k)$ by exact joint diagonalization at the last step of the scheme.
\begin{proposition}
\label{prop:esprit_exist}
Let a window $\Ls$ can be chosen such that the forms $\Ms_x$ and
$\Ms_y$ given in \eqref{eq:Msx} be nonempty and $\rank_{\Ms_x} \tS = \rank_{\Ms_y} \tS = \rank_\Ls \tS = r$.
Denote $\{\tU_k\}_{k=1}^r$,  $\tU_k\in \spaceR^\Ls$, a basis of $\sfS^\Ls(\tS)$.
Then \\
(1) the equalities $\bfP_x \bfM_x = \bfQ_x$ and $\bfP_y \bfM_y = \bfQ_y$ (see \eqref{eq:ESPRITapprox})
have unique exact solutions
$\bfM_x$ and $\bfM_y$;\\
(2) there exist such a matrix $\bfT \in \spaceC^{r \times r}$ that
$\bfM_x = \bfT \diag(\mu_1, \dots, \mu_r) \bfT^{-1}$
and $\bfM_y = \bfT \diag(\nu_1, \dots, \nu_r) \bfT^{-1}$.
\end{proposition}
\begin{proof}
Take the matrix $\bfT$ as a matrix of transformation of a basis $\{ \tU_k \}_{k=1}^r$ to the basis
$\{\tU'_k\}_{k=1}^r$, where
\begin{equation*}
  (\tU'_k)_{l, n} = \mu_k^l\nu_k^n.
\end{equation*}
We prove both statements of the proposition for $\bfM_x$;
for the other direction the proof is the same.
If we consider the  basis $\tU'_k$ instead of $\tU_k$ and perform Algorithm~\ref{alg:ESPRIT}, then the corresponding shifted matrices
$\bfP'_x$ and $\bfQ'_x$ can be represented as follows:
\begin{equation*}
  \bfP'_x = \bfP_x \bfT, \quad \bfQ'_x = \bfQ_x \bfT,
\end{equation*}
since construction of the shifted matrices is a linear procedure.
Obviously, the matrix $\bfM'_x = \diag(\mu_1, \dots, \mu_r)$ will be an exact solution of~\eqref{eq:ESPRITapprox}, $\bfP'_x \bfM'_x = \bfQ'_x$.
Therefore, $\bfP_x \bfT \bfM'_x = \bfQ_x \bfT$ and hence
  $\bfP_x \bfT \bfM'_x \bfT^{-1}= \bfQ_x$.

Now let us prove that the solution $\bfM_x$ is unique.
The condition of uniqueness for the solutions of the equalities~\eqref{eq:ESPRITapprox} is
$\rank \bfP_x = r$.
As before, $(\tU)_{\Ms_x}$ is the restriction of the shaped array $\tU$ to the shape $\Ms_x$.
Then
\begin{gather*}
  \rank \bfP_x = \dim\sspan\{P_1, \dots, P_r\} = \\ =
  \dim\sspan\{(\tU_1)_{\Ms_x}, \dots, (\tU_r)_{\Ms_x} \} = \\ =
  \dim\sspan\{(\tU'_1)_{\Ms_x}, \dots, (\tU'_r)_{\Ms_x} \} = \\ =
  \dim\sfS^{\Ms_x}(\tS) = \rank\nolimits_{\Ms_x} \tS.
\end{gather*}
\end{proof}

The following proposition provides sufficient conditions
for the correctness of the Shaped 2D ESPRIT algorithm for both
planar and circular versions.

\begin{proposition}
\label{cor:ESPRITsuff}
Let
the shape $\Ls$ contain at least one $(r + 1) \times (r + 1)$ square and
the shape $\Ks = \{\KsE \mid \Ls \oplus \{\KsE\} \subset \Ns\}$ contain at least one
$r \times r$ square.
Then the conditions of Proposition \ref{prop:esprit_exist} are valid.
\end{proposition}
\begin{proof}
Obviously, if the shape $\Ls$ contains an $(r + 1) \times (r + 1)$ square, then the
shapes $\Ms_x$ and $\Ms_y$ contain at least one $r \times r$ square, which is (along with the condition on the shape $\Ks$)
the sufficient condition
for $\rank_{\Ms_x} \tS = \rank_{\Ms_y} \tS = \rank_\Ls \tS = \rank \tS = r$ (see Remark~\ref{rem:suffsquare}).
\end{proof}

\section{Examples}
\label{sec:examples}
In this section we consider two examples, for demonstration of Shaped 2D-ESPRIT and of
cylindric Shaped 2D-SSA. All the examples are reproducible, and the code in \pkg{R} is available on request.
\subsection{Shaped 2D-ESPRIT}
Let us describe the example of application of the shaped version of 2D-ESPRIT.
We take a part of the commonly used image ``Barbara'' $512\times 512$ (see e.g. \cite{McGuire2011Data})
containing the table covered by checkerboard cloth.
However, due to edges of the table, the tablecloth consists of three parts with different views of the texture.
Let us consider this picture to demonstrate how Shaped 2D-SSA together with Shaped 2D-ESPRIT can help to
analyze textures in different complex-shaped areas of the image.

We construct masks for three parts of the tablecloth and consider them separately (see Figure~\ref{fig:barbara_snippets}).
The first two subimages are triangular, while the third one has complex-shaped boundary and is non-convex.
Since the first two subimages are convex, we can choose the window of the same shape as the subimage (that is,
a triangular window), but of smaller size ($4$ times smaller for subimage 1 and $9$ times smaller for subimage 2).
For the third subimage of complex shape with a gap, we choose the circle window shape of radius $10$.
Since the amplitudes of the texture components are not constant, we choose not large windows, which are though sufficient
for extraction of the texture.

Checkerboard form of the texture means that it consists of two sine planar waves that run to different directions.
This means that the main part of texture can be extracted by means of four SSA components, each pair corresponds to a separate angle.
Checkerboard texture is easily seen on the first subimage. It is not evident on the second one and is not seen at all on the third subimage.

Analysis of the eigenarrays helps to indicate which eigenarrays correspond to the texture. For example, the third image contains only one-direction strips
on the eigenarray images.
Therefore we choose four eigentriples ET2--5 for the first and the second subimages and two eigentriples ET2,3 for the third one. Figure~\ref{fig:barbara_eigenvalues} contains the result of
reconstruction (top), where the reconstructions of separate strips are depicted under the full reconstruction (for the subimage 3 we have only
one strip and therefore the picture of the second strip is empty).

\begin{figure}[thpb]
  \centering
  \includegraphics{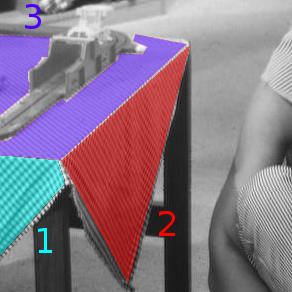}
  \caption{Barbara's table with enumerated subimages}
  \label{fig:barbara_snippets}
\end{figure}

\begin{figure}[thpb]
  \centering
  \includegraphics[width=\columnwidth]{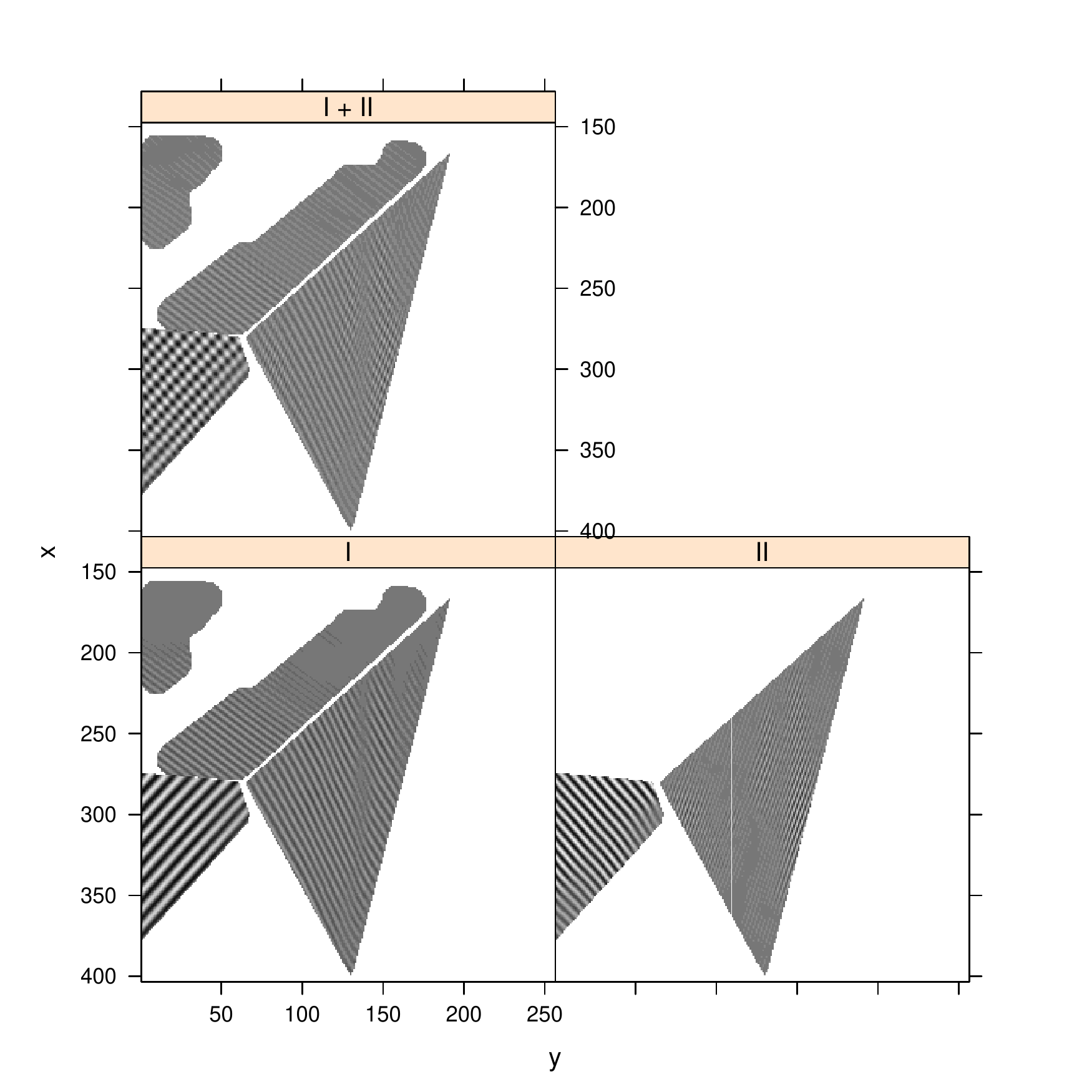}
  \caption{Barbara's table: reconstructions of texture}
  \label{fig:barbara_eigenvalues}
\end{figure}

The frequency estimations obtained by the column LS 2D-ESPRIT are presented in Table~\ref{t:freqs}.
The estimated angles (in degrees) and widths of the strips are also presented. The exponential rates show how rapidly the amplitudes are changed.

The coordinate system is as follows: the $x$ coordinate is measured from top to bottom, while the $y$ coordinate
is measured from left to right.

The relation between periods $t_x$ and $t_y$ allows one to find the angle of strips with the horizontal line by the formula
$\alpha = \arctan(t_x / t_y)$.
The strips on the subimage 1 are approximately orthogonal and have angles near 45$^\circ$. The strips on the subimage 2
are closer to vertical (angles 63$^\circ$ and 70$^\circ$), while that on the image 3 are closer to horizontal.
Note that widths of the strips can be found by the formula $w = t_x \cos\alpha = t_x t_y/\sqrt{t_x^2 + t_y^2}$.

Let us look at the exponential rates. Positive rates correspond to increasing of amplitudes in the direction from the left top corner.
Therefore, for example, the decreasing of the amplitudes of the subimage 3 from left-bottom to right-top means that the $x$-rate is positive,
while the $y$-rate is negative.

\begin{table}[h]
\centering
\caption{Barbara, ESPRIT periods and rates}
\label{t:freqs}
\begin{tabular}{
    ||
    c
    |
    S[table-format=-2.1]
    |
    S[table-format=-2.1]
    |
    S[table-format=-2.1]
    |
    S[table-format=2.0]
    |
    S[table-format=-1.4]
    |
    S[table-format=-1.4]
    ||}
\hline
{$\mathrm{Part}$} & {$t_x$} & {$t_y$} & {$\mathrm{Width}$} & {$\alpha, ^\circ$} & {$\mathrm{Rate}_x$} & {$\mathrm{Rate}_y$} \\
\hline
1    & 8.1      &  6.9     & 5.2   & 50    & -0.0066 & -0.0061 \\
1    & 11.0     &  10.1    & 7.4   & 47    & -0.0022 & -0.0064 \\
2    & 9.7      &  5.0     & 4.4   & 63    & -0.0031 & -0.0075 \\
2    & 7.2      &  2.6     & 2.5   & 70    & 0.0025  & 0.017 \\
3    & 5.1      &  6.8     & 4.1   & 37    & 0.0054  & -0.0046 \\
\hline
\end{tabular}
\end{table}

\subsection{Cylindric Shaped 2D-SSA}
Let us consider the data for expression of gene activity measured on the embryo of the drosophila (fruit fly).
The form of the embryo is similar to ellipsoid and therefore the cylindrical projection is adequate for
a middle part of the embryo. The technique of SSA analysis of such kind of 2D data by planar non-shaped 2D-SSA
was developed in \cite{Golyandina.etal2012a, Holloway.etal2011}.
After projection, the initial data given on 3D embryo surfaces are represented in the form $(x_i, y_i, f_i)$, where $(x_i, y_i)$ are coordinates of nuclei centers in
the cylindric
projection transformed to a planar rectangle,
$f_i$ are the expression values. The aim of the analysis is to decompose the data into the sum of
a pattern and noise and then measure the signal/noise ratio. The data are cropped and
then are interpolated to a regular grid; therefore, we obtain a digital image, which
can be processed by 2D-SSA. After that, the values of the smoothed expression are interpolated backward onto
nuclei centers.

\begin{figure}[thpb]
  \centering
  \includegraphics[width=\columnwidth]{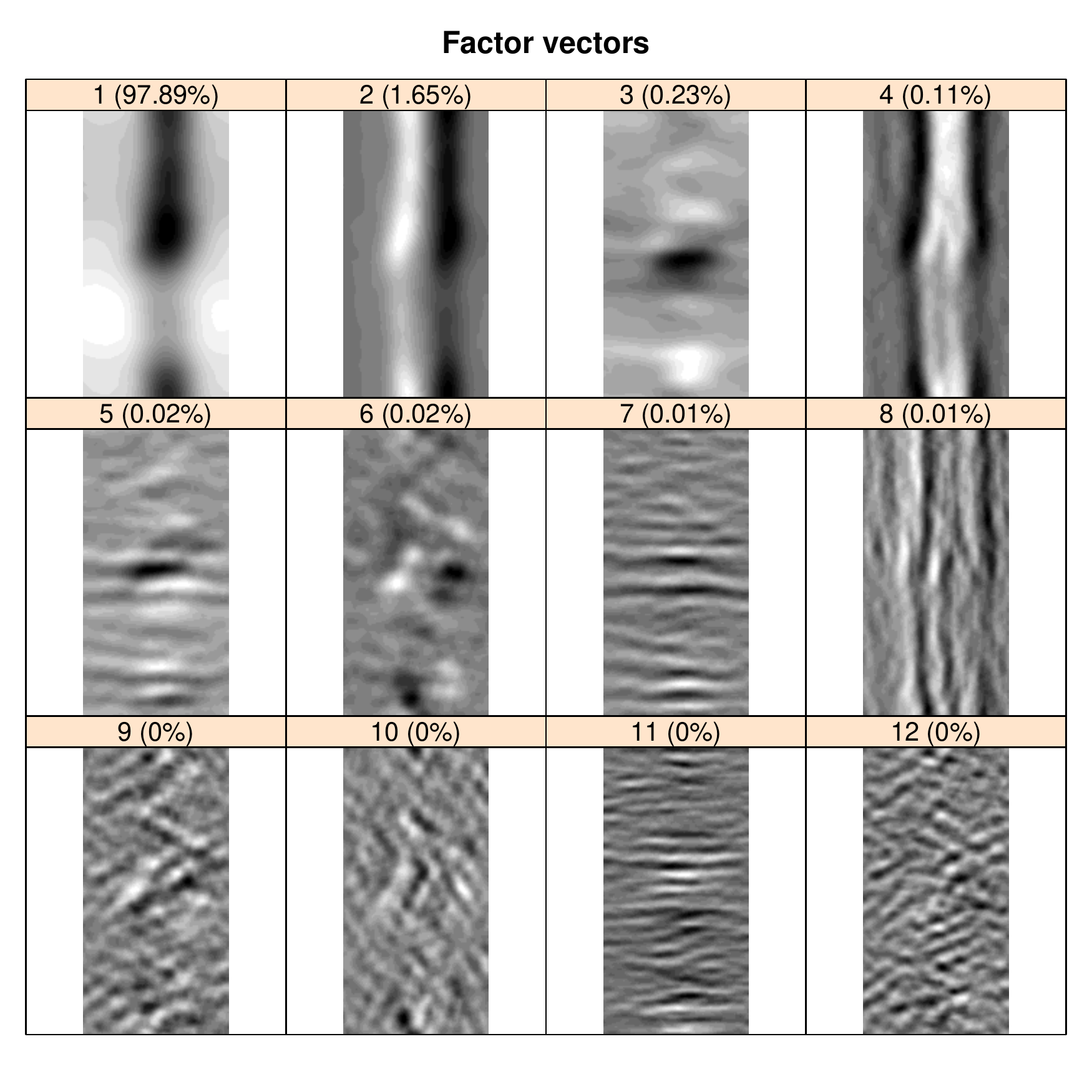}
  \caption{``Kr\"uppel'': factor vectors}
  \label{fig:drosophila_factor}
\end{figure}

\begin{figure}[thpb]
  \centering
  \includegraphics[width=\columnwidth]{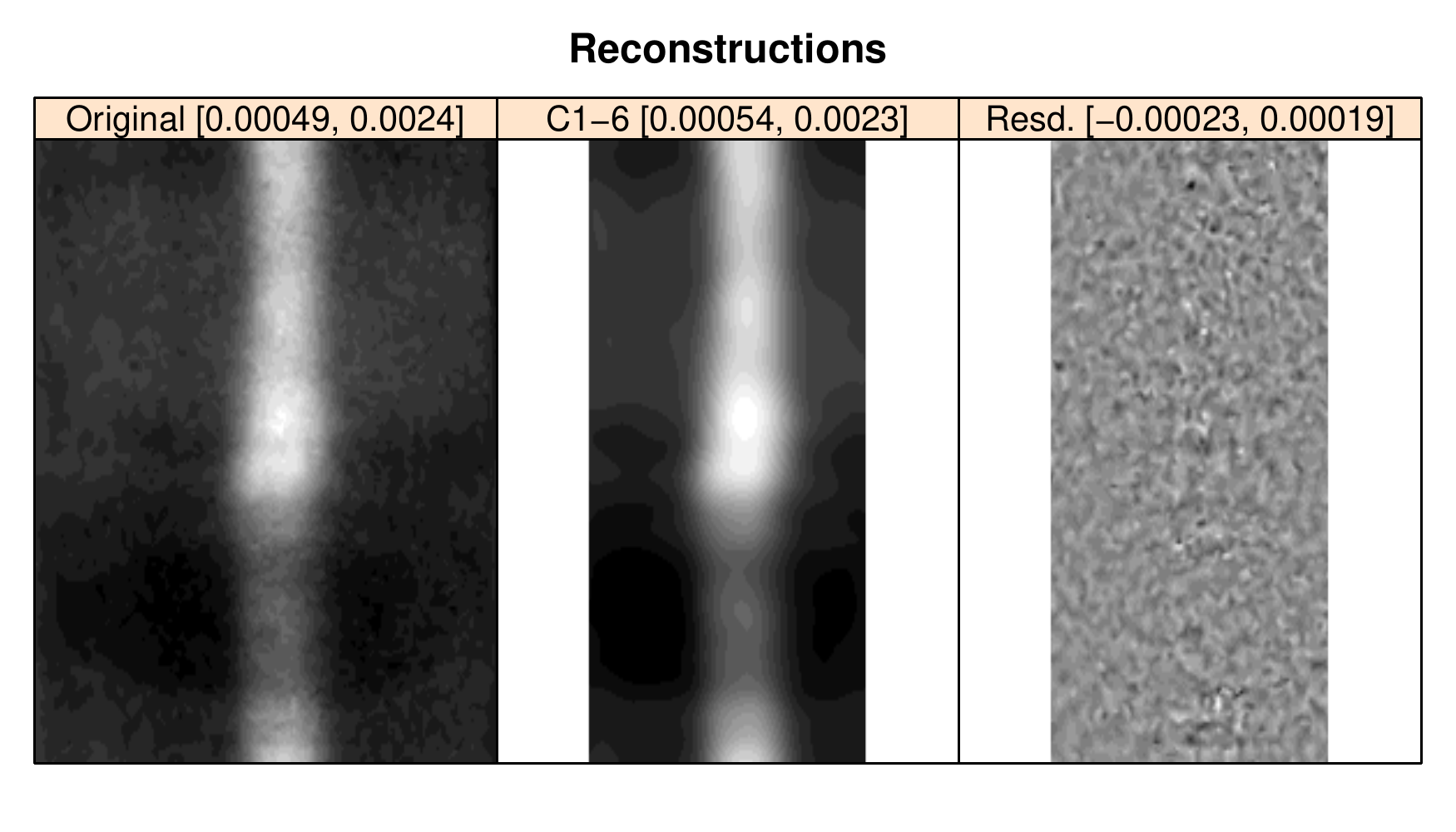}
  \caption{``Kr\"uppel'': original image, reconstruction and residual}
  \label{fig:drosophila_reconstruction}
\end{figure}

Here we demonstrate that the developed circular version of 2D-SSA can perform smoothing without
artificial transformation of the cylinder to a rectangle. Thus the result of processing
does not depend on the technique of this transformation and has no extra edge effects. We consider the cylindrical topology ($T_x = N_x$ and $T_y = \infty$).

The data are downloaded from the BDTNP archive \cite{BDTNP},
the file \texttt{``v5\_s11643-28no06-04.pca''}, gene ``Kr\"uppel''.
The data on the surface of embryo are given in 3D coordinates, therefore we calculated cylindric coordinates
by finding the principal axis of the data. The interpolation is performed with the step 0.5\% and the
middle part from 20 to 80\% of the embryo lengths were processed.

The result of decomposition (the factor vectors from the $12$ leading eigentriples) is depicted in
Figure~\ref{fig:drosophila_factor}, where contributions of the corresponding eigentriples are indicated in the headers.

The components 1--6 grouped together provide an adequate smoothing; the residual oscillates around zero, see Figure~\ref{fig:drosophila_reconstruction}
(the maximal and minimal values of the reconstructed components are shown in the headers).
Note that the bottom and top edges are coincide, that is, there is no edge effect.

\section{Conclusions}
\label{sec:conclusions}
The shaped extensions of SSA are considerably extend the area of the method applications, including
the digital images of non-rectangular shape, images or series with gaps, planar projections of 3D surfaces.
Complemented by the circular case, Shaped SSA is able to deal with cylindrical projections of spherical or elliptic
surfaces and hence to avoid edge effects.

There is a close relation between SSA-like methods and subspace based methods for frequency estimation.
Therefore, new extensions of SSA should produce new extensions of methods for frequency estimation.
In the paper we consider the corresponding ESPRIT-like methods, formulate the algorithms  and the conditions of their correctness.

The algorithms and their implementation are presented for 2D images in the form that can be easily
extended to higher dimensions.



\section*{Acknowledgment}
The authors are thankful to Konstantin Usevich for useful comments and suggestions, which helped to improve
the paper presentation.


\begin{thebibliography}{10}
\providecommand{\url}[1]{#1}
\csname url@rmstyle\endcsname
\providecommand{\newblock}{\relax}
\providecommand{\bibinfo}[2]{#2}
\providecommand\BIBentrySTDinterwordspacing{\spaceskip=0pt\relax}
\providecommand\BIBentryALTinterwordstretchfactor{4}
\providecommand\BIBentryALTinterwordspacing{\spaceskip=\fontdimen2\font plus
\BIBentryALTinterwordstretchfactor\fontdimen3\font minus
  \fontdimen4\font\relax}
\providecommand\BIBforeignlanguage[2]{{%
\expandafter\ifx\csname l@#1\endcsname\relax
\typeout{** WARNING: IEEEtran.bst: No hyphenation pattern has been}%
\typeout{** loaded for the language `#1'. Using the pattern for}%
\typeout{** the default language instead.}%
\else
\language=\csname l@#1\endcsname
\fi
#2}}

\bibitem{Broomhead.King1986}
D.~Broomhead and G.~King, ``Extracting qualitative dynamics from experimental
  data,'' \emph{Physica D}, vol.~20, pp. 217--236, 1986.

\bibitem{Vautard.etal1992}
R.~Vautard, P.~Yiou, and M.~Ghil, ``{S}ingular-{S}pectrum {A}nalysis: A toolkit
  for short, noisy chaotic signals,'' \emph{Physica~D}, vol.~58, pp. 95--126,
  1992.

\bibitem{Elsner.Tsonis1996}
J.~B. Elsner and A.~A. Tsonis, \emph{{S}ingular {S}pectrum {A}nalysis: A New
  Tool in Time Series Analysis}.\hskip 1em plus 0.5em minus 0.4em\relax Plenum
  Press, 1996.

\bibitem{Golyandina.etal2001}
N.~Golyandina, V.~Nekrutkin, and A.~Zhigljavsky, \emph{Analysis of Time Series
  Structure: {SSA} and Related Techniques}.\hskip 1em plus 0.5em minus
  0.4em\relax Chapman\&Hall/CRC, 2001.

\bibitem{Ghil.etal2002}
M.~Ghil, R.~M. Allen, M.~D. Dettinger, K.~Ide, D.~Kondrashov, M.~E. Mann,
  A.~Robertson, A.~Saunders, Y.~Tian, F.~Varadi, and P.~Yiou, ``Advanced
  spectral methods for climatic time series,'' \emph{Rev. Geophys.}, vol.~40,
  no.~1, pp. 1--41, 2002.

\bibitem{Golyandina.Zhigljavsky2012}
N.~Golyandina and A.~Zhigljavsky, \emph{{S}ingular {S}pectrum {A}nalysis for
  time series}, ser. Springer Briefs in Statistics.\hskip 1em plus 0.5em minus
  0.4em\relax Springer, 2013.

\bibitem{Markovsky2012}
I.~Markovsky, \emph{Low Rank Approximation: Algorithms, Implementation,
  Applications}.\hskip 1em plus 0.5em minus 0.4em\relax Springer, 2012.

\bibitem{Roy.Kailath1989}
R.~{Roy} and T.~{Kailath}, ``{ESPRIT: estimation of signal parameters via
  rotational invariance techniques},'' \emph{IEEE Trans. Acoust.}, vol.~37, pp.
  984--995, 1989.

\bibitem{Barkhuijsen.etal1987}
H.~Barkhuijsen, R.~de~Beer, and D.~van Ormondt, ``Improved algorithm for
  noniterative time-domain model fitting to exponentially damped magnetic
  resonance signals,'' \emph{J. Magn. Reson.}, vol.~73, pp. 553--557, 1987.

\bibitem{VanHuffel.etal1994}
S.~Van~Huffel, H.~Chen, C.~Decanniere, and P.~van Hecke, ``Algorithm for
  time-domain {NMR} data fitting based on total least squares,'' \emph{J. Magn.
  Reson. Ser. A}, vol. 110, pp. 228--237, 1994.

\bibitem{Golyandina.Usevich2010}
N.~Golyandina and K.~Usevich, ``{2D}-extension of singular spectrum analysis:
  algorithm and elements of theory,'' in \emph{Matrix Methods: Theory,
  Algorithms and Applications}, V.~Olshevsky and E.~Tyrtyshnikov, Eds.\hskip
  1em plus 0.5em minus 0.4em\relax World Scientific Publishing, 2010, pp.
  449--473.

\bibitem{Golyandina.etal2013}
\BIBentryALTinterwordspacing
N.~Golyandina, A.~Korobeynikov, A.~Shlemov, and K.~Usevich, ``Multivariate and
  {2D} extensions of singular spectrum analysis with the {Rssa} package,''
  \emph{Arxiv.org}, 2013. [Online]. Available:
  \url{http://arxiv.org/abs/1309.5050}
\BIBentrySTDinterwordspacing

\bibitem{Rouquette.Najim01Estimation}
S.~Rouquette and M.~Najim, ``Estimation of frequencies and damping factors by
  two-dimensional {ESPRIT} type methods,'' \emph{IEEE Transactions on Signal
  Processing}, vol.~49, no.~1, pp. 237--245, 2001.

\bibitem{monadjemi04-towards}
A.~Monadjemi, ``Towards efficient texture classification and abnormality
  detection,'' Ph.D. dissertation, University of Bristol, 2004.

\bibitem{trickett08conf-f}
S.~Trickett, ``F-xy cadzow noise suppression,'' in \emph{78th Annual
  International Meeting, SEG, Expanded Abstracts}, 2008, pp. 2586--2590.

\bibitem{Holloway.etal2011}
D.~M. Holloway, F.~J.~P. Lopes, L.~da~Fontoura~Costa, B.~A.~N. Traven\c{c}olo,
  N.~Golyandina, K.~Usevich, and A.~V. Spirov, ``Gene expression noise in
  spatial patterning: \emph{hunchback} promoter structure affects noise
  amplitude and distribution in \emph{Drosophila} segmentation,'' \emph{PLoS
  Computational Biology}, vol.~7, no.~2, p. e1001069, 2011.

\bibitem{Shin.etal13MRiM-Calibrationless}
P.~J. Shin, P.~E.~Z. Larson, M.~A. Ohliger, M.~Elad, J.~M. Pauly, D.~B.
  Vigneron, and M.~Lustig, ``Calibrationless parallel imaging reconstruction
  based on structured low-rank matrix completion,'' \emph{Magnetic Resonance in
  Medicine}, pp. n/a--n/a, 2013.

\bibitem{Korobeynikov.etal2013}
\BIBentryALTinterwordspacing
A.~Korobeynikov, A.~Shlemov, K.~Usevich, and N.~Golyandina, \emph{{Rssa}: A
  collection of methods for singular spectrum analysis}, 2014, {R}~package
  version~0.11. [Online]. Available:
  \url{http://CRAN.R-project.org/package=svd}
\BIBentrySTDinterwordspacing

\bibitem{Golyandina.Korobeynikov2013}
N.~Golyandina and A.~Korobeynikov, ``Basic {S}ingular {S}pectrum {A}nalysis and
  forecasting with {R},'' \emph{Computational Statistics \& Data Analysis},
  vol.~71, p. 1934–954, 2014.

\bibitem{Shlemov.Golyandina2014}
N.~Golyandina and A.~Shlemov, ``Variations of singular spectrum analysis for
  separability improvement: {N}on-orthogonal decompositions of time series,''
  \emph{Stat Interface}, 2014 (accepted).

\bibitem{Mourrain.Pan2000}
B.~Mourrain and V.~Y. Pan, ``Multivariate polynomials, duality, and structured
  matrices,'' \emph{Journal of Complexity}, vol.~16, no.~1, pp. 110--180, 2000.

\bibitem{Korobeynikov2010}
A.~Korobeynikov, ``Computation- and space-efficient implementation of {SSA},''
  \emph{Statistics and Its Interface}, vol.~3, no.~3, pp. 357--368, 2010.

\bibitem{Golub.VanLoan1996}
G.~H. Golub and C.~F. Van~Loan, \emph{{Matrix computations (3rd ed.)}}.\hskip
  1em plus 0.5em minus 0.4em\relax Baltimore, MD, USA: Johns Hopkins University
  Press, 1996.

\bibitem{Korobeynikov2013}
\BIBentryALTinterwordspacing
A.~Korobeynikov, \emph{svd: Interfaces to various state-of-art {SVDs} and
  eigensolvers}, 2014, {R}~package version~0.3.3--2. [Online]. Available:
  \url{http://CRAN.R-project.org/package=svd}
\BIBentrySTDinterwordspacing

\bibitem{Golyandina2010}
N.~Golyandina, ``On the choice of parameters in singular spectrum analysis and
  related subspace-based methods,'' \emph{Stat. Interface}, vol.~3, no.~3, pp.
  259--279, 2010.

\bibitem{Kurakin.etal95JoMS-Linear}
V.~Kurakin, A.~Kuzmin, A.~Mikhalev, and A.~Nechaev, ``Linear recurring
  sequences over rings and modules,'' \emph{Journal of Mathematical Sciences},
  vol.~76, no.~6, pp. 2793--2915, 1995.

\bibitem{Wang.etal05IToSP-Comments}
Y.~Wang, J.-W. Chan, and Z.~Liu, ``Comments on ``estimation of frequencies and
  damping factors by two-dimensional esprit type methods'','' \emph{IEEE
  Transactions on Signal Processing}, vol.~53, no.~8, pp. 3348--3349, 2005.

\bibitem{McGuire2011Data}
\BIBentryALTinterwordspacing
M.~Mc{G}uire. Computer graphics archive. [Online]. Available:
  \url{http://graphics.cs.williams.edu/data/images.xml}
\BIBentrySTDinterwordspacing

\bibitem{Golyandina.etal2012a}
N.~E. Golyandina, D.~M. Holloway, F.~J. Lopes, A.~V. Spirov, E.~N. Spirova, and
  K.~D. Usevich, ``Measuring gene expression noise in early drosophila embryos:
  Nucleus-to-nucleus variability,'' in \emph{Procedia Computer Science},
  vol.~9, 2012, pp. 373--382.

\bibitem{BDTNP}
\BIBentryALTinterwordspacing
Berkeley drosophila transcription network project. [Online]. Available:
  \url{http://bdtnp.lbl.gov/Fly-Net/}
\BIBentrySTDinterwordspacing

\end{thebibliography}
\end{document}